\pdfoutput=1
\newif\ifFull
\Fulltrue

\documentclass[runningheads]{llncs}
\usepackage[T1]{fontenc}

\usepackage{cite}
\usepackage{amsmath,amssymb,amsfonts}
\usepackage{algorithm}
\usepackage[noend]{algorithmic}
\usepackage{graphicx}
\usepackage{textcomp}
\usepackage{cleveref}
\usepackage{xcolor}
\usepackage{comment}
\def\BibTeX{{\rm B\kern-.05em{\sc i\kern-.025em b}\kern-.08em
    T\kern-.1667em\lower.7ex\hbox{E}\kern-.125emX}}

\newtheorem{thm}{Theorem}

\renewenvironment{proof}{\noindent{\bf Proof:}}{\hspace*{\fill}\rule{6pt}{6pt}\medskip}
\ifFull\else
\renewcommand{\subsection}[1]{\paragraph{\textbf{#1}.}}
\fi
\newcommand{\bemph}[1]{\textbf{\textit{#1}}}

\usepackage[]{todonotes} %disable option
\definecolor{okabe1}{HTML}{000000}
\definecolor{okabe2}{HTML}{E69F00}
\definecolor{okabe3}{HTML}{56B4E9}
\definecolor{okabe4}{HTML}{009E73}
\definecolor{okabe5}{HTML}{F0E442}
\definecolor{okabe6}{HTML}{0072B2}
\definecolor{okabe7}{HTML}{D55E00}
\definecolor{okabe8}{HTML}{CC79A7}

\newcommand{\ham}{{\mathrm{hd}}}
\newcommand{\rh}{{\mathrm{rh}}}
\newcommand{\slide}{{\mathrm{Slide}}}
\DeclareMathOperator{\polylog}{polylog}
\newcommand{\no}[1]{}

\begin{document}

\pagestyle{plain}

\title{Simple Low-Overhead Communication-Efficient String Reconciliation
       and Edit Distance\thanks{This research supported in part by NSF grant 2212129, by Basal Funds FB0001 and ABF240001, ANID, Chile, and by Fondecyt Grant 1260080, ANID, Chile.}}

%\author{Anonymous author(s)}
\begin{comment}
\author{\IEEEauthorblockN{Michael T. Goodrich}
\IEEEauthorblockA{\textit{Dept. of Computer Science} \\
\textit{Univ. of California, Irvine}\\
Irvine, CA USA\\
goodrich@uci.edu}
\and
\IEEEauthorblockN{Gonzalo Navarro}
\IEEEauthorblockA{\textit{Dept. of Computer Science} \\
\textit{Univ. of Chile}\\
Santiago, Chile\\
gnavarro@dcc.uchile.cl}
\and
\IEEEauthorblockN{Claire A. To}
\IEEEauthorblockA{\textit{Dept. of Computer Science} \\
\textit{Univ. of California, Irvine}\\
Irvine, CA USA\\
claire.to@uci.edu}
}
\end{comment}

\author{Michael T. Goodrich\inst{1}\orcidID{0000-0002-8943-191X} \and
Gonzalo Navarro\inst{2}\orcidID{0000-0002-2286-741X} \and
Claire A. To\inst{1}\orcidID{0009-0008-9102-2219}}
%
% \authorrunning{G. Author et al.}
\institute{University of California, Irvine, Irvine CA, USA
% \and
% University of Chile, 
% \url{http://www.springer.com/gp/computer-science/lncs} 
\and
CeBiB \& Dept. of Computer Science, University of Chile, Santiago, Chile}

\maketitle

\begin{abstract}
Suppose two parties, Alice and Bob, hold long character strings, $X$ and $Y$, respectively, and they are interested in determining how similar $X$ and $Y$ are. {Moreover, they want to exchange the strings with cost proportional to their degree of dissimilarity.} 
\ifFull
Such problems arise, for example, in database and file system synchronization operations, as well
as in DNA sequence comparisons.
Since the strings are long, we are interested in methods that are communication-efficient and have low overhead in terms of the computations that Alice and Bob must perform, when the strings are similar enough.
\fi
In this paper, we provide simple low-overhead communication-efficient algorithms for such string reconciliation and edit distance
problems.
In the general case, 
%where the only assumption we make is that we have an upper bound, $k$, on the edit distance between $X$ and~$Y$, 
we show how to determine the edit distance {$k$} between $X$ and~$Y$ using only $O(k^2\log n)$ bits of communication
and optimal $O(n)$ time overhead, with high probability.
For specialized cases, such as typical English text or DNA sequences, where we can make additional well-justified assumptions
about the distribution of the input strings, we show how to achieve possibly better bounds, such as  $O(k\log^5 n)$ bits of communication.


\end{abstract}

% \begin{IEEEkeywords}
% distributed algorithms, randomized algorithms,
% edit distance, string algorithms, hash functions, 
% communication complexity, string reconciliation 
% \end{IEEEkeywords}

\section{Introduction}

A fundamental challenge in distributed computing lies in managing distributed data across communication networks. A problem arising frequently is that of checking or maintaining consistency across records in replicated databases, synchronizing files in peer-to-peer systems, and repairing or updating  databases across the network. Thus, there is a need for efficient distributed protocols for reconciling character strings, which can, e.g.,
represent DNA sequences, avoiding the communication overhead of retransmitting the whole strings.
% Probably handled with CVS systems, not convincing...
%This problem is most apparent in cloud-based storage systems, where millions of users simultaneously create, edit, and share documents and other character strings across multiple devices: most updates are incremental and consist of small changes; hence, using full-transfer synchronization methods would be inefficient in terms of  bandwidth and computational overhead. In another scenario, 
\ifFull
As a concrete scenario, sequencing labs repeatedly generate genomes, and their clients then download them over the Internet to update genome libraries. In such scenarios, where clients already hold other genomes of the same species, which are known to be very similar to the new ones, clients and servers would like to minimize communication complexities as well as computational overheads. 
\fi
%For example, distributed DNA databases must ensure availability, fault tolerance, and consistency. 
%An issue that arises is that transmitting entire strings or large natural-language documents for every update is inefficient. That is, 

% Similar challenges arise in web server mirroring, large-scale data distribution, and even image or multimedia reconciliation, a two dimensional generalization of string reconciliation. Beyond computing, string reconciliation underpins problems in bioinformatics, including DNA sequencing, protein reconstruction, and reliably matching biometric data.

\ifFull
In some cases, version control systems can keep track of where the differences are, or both parties can share a common string, so that they can easily exchange only the needed edits to reconcile the strings. The scenarios we have described, however, may lack such shared string or edit information, so the parties {\em know} that the differences between their strings are small, but do not know {\em which} they are. Still, they want to find those differences, and reconcile the strings, within a communication cost that depends only on the amount of differences.
\fi
%and would like to determ
%Accordingly, by exploiting the similarities between versions of large documents and between DNA sequences of instances of the same species, systems can ideally transmit only what has changed, so as to dramatically reduce communication. 
%Leveraging this characteristic not only accelerates document synchronization in large-scale distributed systems but also extends to vast applications, from cloud collaboration to computational biology, document processing, and secure data exchange. 
%Still, reducing communication complexity at the cost of high computational overhead for the communicating parties is undesirable, since computation time is also a finite resource that the parties want to optimize.
%Thus, studying the problem of efficient string reconciliation, in terms of both computational overhead and communication complexity, is well-motivated for modern data-driven distributed applications.

\subsection{Problem Statement}
Let $X = x_1 x_2 \ldots x_l$ and $Y = y_1 y_2 \ldots y_n$ be two strings held by two parties, Alice and Bob, respectively, 
where each $x_i, y_j \in \Sigma$, for $i=1,2,\ldots, l$ and $j=1,2,\ldots,n$, and $\Sigma$ is an alphabet whose size
may or may not depend on~$l$ or~$n$. 
In the case where $l=n$,
the {\em Hamming distance}, $\ham(X,Y)$, between $X$ and~$Y$ is 
the number of indices, $i$, where $x_i\not=y_i$, that is,
the number of character substitutions needed to transform $X$ into~$Y$
\cite{navarro2001guided,crochemore2007algorithms}.
A related notion that does not assume $l=n$
is the \bemph{edit distance}, also
known as {\em Levenshtein distance}, $L(X,Y)$,
between $X$ and~$Y$: it is 
the minimum number character insertions, deletions, or substitutions
needed to transform $X$ into~$Y$.

In this paper, we are interested in distributed string
reconciliation algorithms with optimal (i.e., linear) computational 
overhead and small communication complexity 
for computing the edit distance between the two strings,
$X$ and $Y$, that are respectively held by Alice and Bob.
{This turns out to be equivalent to the problem of Alice and Bob exchanging their strings $X$ and $Y$ with communication cost based on their communicating a minimal set of edits between $X$ and $Y$.}
We assume that Alice and Bob can communicate over some channel but
they do not share any other computational
resources, such as CPUs or memory.
The \bemph{communication complexity} for a communication protocol for Alice and Bob is the number of bits that are communicated between them over the course of their distributed protocol, where we assume initially that neither party has knowledge of the other party's string.
% See Figure~\ref{fig:problem}.

% \begin{figure}[hbt]
% \centering
% \includegraphics[width=.85\columnwidth]{figs/Problem.png}
% \caption{\label{fig:problem} The distributed edit distance problem.
% Image adapted from an image generated using Nano Banana Pro.}
% \end{figure}

% GZL: Removed this additional example
%For example, $X$ may be a disk image prior to updating a handful of software applications and $Y$ is a similar disk image after the updates. To distribute the software updates most efficiently, we would like to determine the edit distance between $X$ and $Y$, and then just send the changes to bring $X$ up to date. Moreover, as mentioned above, we want to do this with minimal computational overhead for the communicating parties and with small communication complexity.

In the \bemph{general distributed edit distance} problem, 
%the only assumption we make is that we are given a parameter, $k\ge1$, such that $L(X,Y)\le k$, and 
we are interested in a protocol for Alice and Bob to determine the edit
distance, {$k = L(X,Y)$}, between $X$ and $Y$ using a small number of bits of communication {(as a function of $k$)}.
Alternatively, in \bemph{specialized} string reconciliation problems, we make additional assumptions regarding
$X$ and $Y$ and their edit distance.
For example, we develop efficient distributed protocols for the case where $X$ and~$Y$ come from well-motivated real-world domains,
such as English text and DNA sequences, which support assumptions regarding $X$ and $Y$ that we can
exploit.
Further, for practical reasons, in this paper we are interested in simple algorithms that have low communication
complexity and computational overhead for Alice and Bob, for both the general and specialized versions 
of the string reconciliation problem.

Our results will hold \bemph{with high probability} (whp), meaning that they occur
with probability $1-1/n^c$, for some constant $c\ge1$. W.l.o.g.\ let us assume $l \le n$.
Our protocols use various independent sources of randomness, e.g., rolling-hash fingerprints, CDC boundary detection, chunk checksums, and IBLT hashing, each of which satisfies the whp guarantee.

\subsection{Our Contributions}
% GZL repetitive
%In this work, we present an approach reconciling strings or sequences based on viewing differences as character-level edits. By leveraging the high similarity between sequences, such as human genomes, which are 99.9\% identical, or disk images after a small number of software updates, we develop efficient techniques to compute and communicate string differences.
%
For the general distributed edit distance problem, we present a simple algorithm that achieves a communication
complexity of {$O(k^2\log n)$} bits
and optimal $O(n)$ time overhead, where $k$ is 
%an upper bound on 
{the edit distance, 
with a probability of success over $1-1/n^2$}.
Our method is a communication-efficient adaptation of a dynamic-programming table ``sliding'' algorithm of Landau, Myers, and Schmidt~\cite{landau1998incremental} {combined with Ukkonen's technique \cite{Ukk85}}. %; see also, e.g.,~\cite{chakraborty2016streaming}.
We avoid performing full sliding steps as in these classical algorithms, however, which
 would be communication inefficient.
Instead, we perform what we call ``fast sliding'' steps by combining rolling hash functions and noisy binary searching.

We also show how to solve string reconciliation in the specialized context, as described.
%that come from real-world applications where the strings, $X$ and $Y$, held respectively by Alice and Bob come from well-known distributions, such as in natural language applications or genomic applications. %, where strings {follow some well-known distributions}. %have a reasonable amount of entropy.
For such scenarios, we provide an algorithm that uses $O(n)$ time and a communication
complexity of
   \[
    O(k(d_X+d_Y)^2\log n((d_X+d_Y+\log n)\log |\Sigma| +\log \eta \log n))
    \]
    %$O(k(d(T_X)+d(T_Y)+\log^2 n))$ bits, 
    bits, succeeding with probability $1-1/\eta$,
where $k$ is an upper bound on the edit distance between $X$ and $Y$ {known a priori},
and $d_X$ and $d_Y$ are the lengths of the longest repeated substrings within $X$ and $Y$, which are expected to be small for real-world
applications {and are $O(\log n)$ on widely used 
%GZL: log base sigma only on uniformly distributed text,
% which is too restrictive. O(log n) is much more general
statistical models. Under such models, the communication complexity is $O(k\log^5 n)$ bits if $|\Sigma|=O(n^c)$ and the probability of error is $O(1/n^c)$ for any constant $c$}. With $O(n\log k)$ time overhead we obtain the same communication cost, now with $k$ being the actual
edit distance, without need to know it a priori.
Our method involves a non-trivial use of the invertible Bloom lookup table (IBLT) data structure~\cite{goodrich2015invertiblebloomlookuptables,eppstein2011whatsthedifference,eppstein2010straggler}.

% In addition, we provide a simple algorithm for performing string reconciliation for strings, such as DNA strings,
% that are stored in compressed form as a list of edits with respect to a common reference string.
% Our method for this application also uses IBLTs, and it runs in time proportional to the sizes of the compressed strings
% using a communication complexity of $O(k\log^2 n)$ bits, with high probability, where $k$ is $1$~plus the edit distance.

%We report in Appendix~\ref{app:exper} on some experimental evaluations for our algorithms, which show that they are efficient in practice.

\section{Related Work}
There has been considerable prior work on the general distributed edit distance problem, but we are not aware of 
any previous methods that are simple and communication-efficient while also having low computational overhead, which refers to the total number of local RAM operations performed by Alice and Bob, excluding time spent transmitting bits over the communication channel.
For example, Orlitsky~\cite{orlitsky1991interactive} gives a communication-optimal method that achieves
$O(k\log n)$ communication complexity, where $k$
is an upper bound on the edit distance between
two strings of size $\Theta(n)$, but requires computational overhead that is $n^{O(k)}$, i.e., exponential
in $k$.
The computational overhead for such protocols has been subsequently 
improved~\cite{jowhari2012efficient,chakraborty2016streaming,irmak2005improved,belazzougui2015efficient},
but the best methods still require $O(n\polylog(n))$ overhead and are still fairly complicated.
For instance, Chakraborty, Goldenberg, and Kouck{\`y}~\cite{chakraborty2016streaming} present a protocol 
with communication efficiency of $O(k^2\log n)$ bits and $O(n\polylog(n))$ overhead by embedding strings in 
a Hamming space using random walks and doing their communication in this Hamming space.
Belazzougui and Zhang~\cite{focs} achieve a communication complexity of $O(k(\log^2 k +\log n))$ bits
and $O(n\polylog(n))$ overhead assuming $k<n^{1/c}$ for a large constant $c$, but
their method is quite complex and their bounds have large constant factors.
In general, the previous methods for general distributed edit distance~\cite{jowhari2012efficient,chakraborty2016streaming,irmak2005improved,belazzougui2015efficient,focs} all
require $O(n\polylog(n))$ overhead, or worse, which in practice is inefficient for the long strings that arise in various real-world settings. % involving English text, DNA sequences, or disk images.

For specialized string reconciliation,
Agarwal, Chauhan, and Trachtenberg~\cite{puzzles}
present distributed methods for reconciling typical real-world strings. Their method has communication
complexity that scales linearly in the edit distance between the reconciling 
strings, as is also true for our methods.
Their reconstruction methods are based on using masks and shingling
to encode the strings and enumerating Eulerian paths to perform the
reconciliation.
As they admit, however, the number of such paths (and, hence, the
computational burden of such algorithms) can be exponentially large, hence, they do not achieve low overhead.
Like our approach, Kontorovich and Trachtenberg~\cite{Kontorovich}
reduce string reconciliation to 
the set reconciliation problem, but their approach still has suboptimal
performance. 
% For example, even for certain random strings and in ideal
% cases, their approach requires $O(L(X,Y)\log^2 n)$ communication complexity to
% reconcile two length-$O(n)$ strings, $X$ and $Y$, 
% having edit distance $L(X,Y)$.
Song and Trachtenberg~\cite{song} revisit the approach using masks
and shingles with an improved Eulerian-path reconstruction method,
which they show emprically improves over the prior work
by Agarwal, Chauhan, and Trachtenberg~\cite{puzzles}, but the asymptotic
overhead of their method is unfortunately still not efficient.

\section{A General Algorithm for Edit Distance}
% \section{Preliminaries}
% Let us begin with some preliminaries.
%GZL Already said 
%For the string reconciliation problems we study in this paper, let us use $l$ to denote the size of the string $X$, and $n$ to denote the size of the string $Y$, where, without loss of generality, $l\le n$.
% Also, 

In this section, we provide a general communication-efficient
distributed algorithm for Alice and Bob to compute the edit distance
of their respective strings, $X=x_1x_2\ldots x_l$ and $Y=y_1y_2\ldots y_n$, using $O(k^2\log n)$ bits of
communication and optimal $O(n)$ time overhead, {where $k$ is 
%a given upper bound on 
the edit distance between $X$ and $Y$}. 
Note that in this case we can also assume, without loss of generality, that $l$ is $\Theta(n)$ and $k<\sqrt{n/\log_{|\Sigma|}n}$, 
since if this is not the case, Alice and Bob can simply exchange $X$ and $Y$
with each other, with $O(n\log|\Sigma|) \subseteq O(k^2\log n)$ bits of communication, and each can privately compute $L(X,Y)$ via the chosen exact edit-distance algorithm: {Our algorithm can start with the $O(k^2\log n)$-bits plan and switch to this second option 
%when it detects that $k \ge \sqrt{n/\log_{|\Sigma|}n}$, 
right  after having communicated $O(n\log |\Sigma|)$ bits.} 
%Indeed, to simplify our discussion, let us assume that $l=n$, since the more general case follows immediately.

We first review the relevant hash functions and recall the classic sliding algorithm. We then present our distributed version 
using $O(k^2\log^2 n)$ bits of communication, and finally 
show how to reduce it to $O(k^2\log n)$. 

\subsection{Hash Functions} \label{app:hashing}
In this paper, we assume that hash functions are pseudo-random and can be considered as random functions for the sake of analysis.
One particular type of hash function we use is a \bemph{rolling hash function},
which can be computed for substrings of a string, $X=x_1x_2\ldots x_l$.
In particular, the well-known Rabin-Karp string pattern matching
algorithm~\cite{rabin-karp} uses the following rolling hash function,
for hashing the substring, $x_ix_{i+1}\ldots x_j$, in which we view each $x_i$ 
as an integer:
\[
\rh(x_ix_{i+1}\ldots x_j) = x_i b^{j-i}+x_{i+1} b^{j-i-1}+\cdots +x_j b^0.
\]
Here, all arithmetic is done on a suitably chosen finite field {(e.g., modulo some prime $p$)} and $b$
is a randomly chosen generator for that field.
We also use this rolling hash function, but rather than using it only for
fixed-length substrings, as in the Rabin-Karp algorithm, we use it in some cases for fixed-length substrings
and other times for arbitrarily long substrings.
One useful property of this rolling hash function is that we 
can compute the rolling hash for every prefix of a length-$l$ string, $X$, in $O(l)$
time using the following inductive formula for each prefix:
\[
\rh(x_1x_{2}\ldots x_{i})=\rh(x_1x_2\ldots x_{i-1})b+x_i.
\]
Then, if we store
these hash values for every prefix of $X$,  
we can compute the value of $\rh(x_ix_{i+1}\ldots x_j)$ for
any substring, $x_ix_{i+1}\ldots x_j$, in constant time as follows:
$$\rh(x_ix_{i+1}\ldots x_j)=
\rh(x_1x_2\ldots x_j)-
{\rh(x_1x_2\ldots x_{i-1}) b^{j-i+1}}\!\!.$$
%\begin{eqnarray*}
%\rh(x_ix_{i+1}\ldots x_j)&=&\big(\rh(x_1x_2\ldots x_j)-
%\rh(x_1x_2\ldots x_{i-1})\big) \\
%                         & & \cdot~ (b^{-1})^{i-1}.
%\end{eqnarray*}
Alternatively, if we are interested in computing rolling hash values
of substrings that have a fixed-length, $r$, 
then we can use the following formula
to compute the next hash value from the previous one in constant time:
% assuming we have stored the value of $b^r$:
%\begin{eqnarray*}
$$\rh(x_ix_{i+1}\ldots x_{i+r-1}) = \rh(x_{i-1}x_{i}\ldots x_{i+r-2})b 
       + x_{i+r-1} - x_{i-1}b^{r}.$$
%\end{eqnarray*}

\subsection{The Classic Sliding Algorithm}
%Our algorithm, which we call the ``fast-sliding algorithm,'' is a communication-efficient adaptation of
We first review a classic ``sliding'' algorithm to compute edit distance by Landau, Myers, and Schmidt~\cite{landau1998incremental}, {combined with Ukkonen's technique \cite{Ukk85}}.
%; see also, e.g.,~\cite{chakraborty2016streaming}.
Let us begin describing the classic dynamic programming algorithm
for computing edit distance~\cite{wagner1974string}. 
Let $L(i,j)$ denote the edit distance for the prefixes $X_i=x_1x_2\ldots x_i$ 
of $X$ and $Y_j=y_1y_2\ldots y_j$ of $Y$.
Then, we have the boundary cases $L(i,0)=i$ and $L(0,j)=j$,
and, in general, the recurrence:
\[
L(i, j) = \min
\begin{cases}
    L(i-1, j) + 1 & \text{(deletion} \text{)} \\
    L(i, j-1) + 1 & \text{(insertion} \text{)} \\
    L(i-1, j-1) + \delta(x_i,y_j), 
\end{cases}
\]
where $\delta(x_i,y_j)= 0$ if $x_i=y_j$ (match) and 
$\delta(x_i,y_j)= 1$ if $x_i\not=y_j$ (substitution).
This set of equations defines an $l\times n$ matrix, $L$, where
$L(l,n)$ is the edit distance between $X$ and $Y$, and the sequence
of values that lead to this value can be viewed as a path that starts
at $L(0,0)$ and follows horizontally rightward, vertically downward, or 
down-right diagonal edges
between adjacent cells.
% See Figure~\ref{fig:edit}.

% \begin{figure}[hbt]
% \centering
% \includegraphics[width=.85\columnwidth]{figs/Edit2.pdf}
% %\includegraphics[width=\columnwidth]{figs/edit.png}
% \vspace*{-18pt}
% \caption{\label{fig:edit} 
% An example edit distance dynamic-programming table,
% with the path leading to the solution highlighted and the waves $h=0$ to $h=5$, the latter including the rightmost-bottom cell, so the edit distance is $k=5$. Note that the wave $h$ spans only diagonals $-h$ to $h$.
% Image adapted from image from stackoverflow user cliff\_leaf; licensed under the 
% Creative Commons Attribution-ShareAlike (CC BY-SA) license.}
% \end{figure}
%
%Further, 
If $L(l,n)\le k$, this path will be restricted
to fall between the $j=i-k$ and $j=i+k$ diagonals of $L(i,j)$ 
entries~\cite{uk}, which
we call the \bemph{central band} of $L$.

As noted by Landau {\em et al.}~\cite{landau1998incremental},
most of the edges in the central band of $L$ go between cells that have the
same values and only $O(k^2)$ of these go between cells whose values differ,
and these differ by $1$.
The key idea in their algorithm is to use suffix trees \cite{Wei73,McC76} to quickly ``slide''
along diagonal paths where values in $L$ are not changing.
%In our case, however, we do not have the luxury of communicating suffix
%trees for $X$ and $Y$, since these have $\Theta(n)$ size.

For a diagonal, $d$, and row, $i$, let $\slide(d,i)$ denote the furthest
row, $\rho\ge i$, that can be reached from row $i$ on diagonal $d$ with edit cost $0$. That is,
\begin{eqnarray*}
\slide(d,i) &=& \max\{\rho\ |\ i\le \rho \le \min\{l,n-d\} \\
            && \mbox{~~~~~~and~}
    x_{i+1}\ldots x_\rho = y_{i+d+1}\ldots y_{\rho+d}\}.
\end{eqnarray*}
Note that 
$\slide(d,i)$ does not compare $x_i$ and $y_i$. Also, note that $\slide(0,0)$ is
the length of the longest common prefix of $X$ and $Y$.
Further, if $x_{i+1}\not= y_{i+d+1}$, then $\slide(d,i)=i$, since
the matching substrings of $X$ and $Y$ 
in the definition of $\slide(d,i)$ in this case
are empty.

In addition, define $M^h(d)$ to be the maximum row index of an entry
on diagonal $d$ that has value $h$, and define the \bemph{$h$-wave} to
be the set
\[
M^h = \{M^h(-h),M^h(-h+1),\ldots,M^h(h-1),M^h(h)\}.
\]
To handle boundary cases, we set $M^h(d)=\infty$ if $M^h(d)$ does not actually exist. 
% See Figure~\ref{fig:edit} again.

At a high level, the algorithm computes the $h$-waves, for $h=0,1,2,\ldots$,
until a wave $h'$ is computed for which 
 $M^{h'}(n-l)=l$,
in which case
the algorithm outputs $h'$ as the edit distance.
The ``inner loop'' formula for computing the $h$-wave from the $(h-1)$-wave,
for $h>0$, is the following~\cite{landau1998incremental}:
\[
M^h(d)=\slide(d,i_{\max}(d,h)),
\]
where
\[
i_{\max}(d,h) = \max \begin{cases}
   M^{h-1}(d+1)+1 & \mbox{if $d<h-1$} \\
   M^{h-1}(d)+1 & \mbox{if $-h<d<h$} \\
   M^{h-1}(d-1) & \mbox{if $d>-h+1$}.
   \end{cases}
\]
Since there are only $O(k)$ entries from the matrix $L$ in each $h$-wave that
we need to consider and only $k$ $h$-waves that we need to consider, this implies
that we just need to perform $O(k^2)$ slide computations to determine
the edit distance. % between $X$ and $Y$.

{
Because $k$ is unknown in advance, we compute the waves incrementally. Following Ukkonen \cite{Ukk85}, we first compute the $0$-wave, for which we need only the diagonal $0$; we then compute the $1$-wave, for which we need only the diagonals $-1$, $0$, and $1$. In general, computing the $h$-wave only requires diagonals $-h$ to $h$. Thus, for the $k$-th wave, we compute
\[
\sum_{h=0}^{k}(2h+1)=O(k^2)
\]
wave entries $M^h(d)$ in order to find the edit distance $k=h'$. 
By computing each entry in $O(1)$ time using the suffix trees of $X$ and $Y$, it is possible to compute the edit distance in time $O(k^2)$.}

\subsection{Our Fast-Sliding Algorithm}
Our algorithm, which we call the ``fast-sliding algorithm,'' is a communication-efficient adaptation of the sliding algorithm; see also, e.g.,~\cite{chakraborty2016streaming}.
In the distributed setting, we cannot afford to use suffix trees 
as done by Landau {\it et al.}, as exchanging them would be too costly.
Our method instead uses a Karp--Rabin rolling hash function, $\rh$.
% (see Appendix~\ref{app:hashing}).
In particular, to compute $\slide(d,i)$, in the general case,
we need to find the maximum matching index, $\rho$, such that 
\[
x_{i+1}\ldots x_\rho = y_{i+d+1}\ldots y_{\rho+d}.
\]
We do this by a doubling-binary search \cite{BENTLEY197682} starting from the index~$i+1$ in
$X$ and index $i+d+1$ in $Y$, where we communicate and compare
$\rh(x_{i+1}\ldots x_{t})$ and 
$\rh(y_{i+d+1}\ldots y_{t+d})$, for $t>i+1$.

We choose the parameters for $\rh$ in this case
so that the hash values have at least $3\log n$ bits, so that with probability
at least $1-n^{-3}$ each test during a doubling-binary search will successfully
determine whether
$x_{i+1}\ldots x_t = y_{i+d+1}\ldots y_{t+d}$ just by comparing the
respective rolling-hash values. 

\begin{lemma}
    The fast-sliding algorithm uses $O(k^2\log^2 n)$ bits of communication and succeeds with high probability.
\end{lemma}
\begin{proof}
    Each doubling-binary search requires $O(\log n)$ such tests, each communicating $O(\log n)$ bits using the rolling hash. 
    Since the classic sliding algorithm performs $O(k^2)$ slide computations, this implies that the entire algorithm has a communication complexity of $O(k^2\log^2 n)$ bits and performs $O(k^2 \log n)$ hash comparisons.
    Under our assumption that $k < \sqrt{n/\log_{|\Sigma|} n}$, the number of hash comparisons is $O(n)$. 
    Each hash comparison fails with probability at most $n^{-3}$, so by a union bound over all comparisons gives an overall failure probability of $O(n^{-2})$.
\end{proof}

\subsection{Improving the Communication Complexity}
We can reduce the communication per comparison to $O(1)$ bits by treating the comparisons during a doubling-binary search as ``noisy'' binary-search queries, similar to a technique used
by Viola~\cite{viola}.
We first compute an $O(\log n)$-bit rolling hash 
\[
rh(\cdot)=rh_1(\cdot),\dots,rh_{L_c}(\cdot)
\]
where each $rh_j$ is a $4$-bit hash component and $L_c\approx 6\log_2 n$.
For the $j$-th comparison, Alice and Bob compare only the $j$-th $4$-bit component.

If two substrings are different, the probability a single $4$-bit component matches is $2^{-4}=1/16$, so each comparison succeeds with probability at least $15/16$, and the result of each comparison is independent of any other. We use the noisy-search procedure of~\cite{dereniowski2025noisy}, also see e.g.,~\cite{feige1994computing,eppstein2025computational}. 

\begin{lemma}\label{lem:4bits}
    For a noisy comparison with error probability $p=1/16$, the noisy-search procedure of~\cite{dereniowski2025noisy} finds the target among $n$ indices with failure probability $\delta=n^{-3}$ using $O(\log n)$ queries, each requiring $O(1)$ bits of communication.
\end{lemma}
\begin{proof}
    By the noisy-search bound from~\cite{dereniowski2025noisy}, the expected number of queries is upper bounded by 
    \[
    \frac{(\log_2 n + \log_2 (\delta^{-1})+3)}{I(p)}
    \]
    where $I(p)=1-H(p)$ are the bits of information gained and $H(p)= -p \log_2 p - (1-p) \log_2 (1-p)$ is the binary entropy function. For $\delta=n^{-3}$ and $p=1/16$, 
    % $I(p)\approx 0.6627$, 
    we have an upper bound of $\approx 6.035 \log_2 n+4.527$ queries, each exchanging a $4$-bit component.
\end{proof}

% We can improve the communication complexity or our fast-sliding algorithm
% by replacing the doubling-binary search with a communication-efficient ``noisy'' doubling-binary search, similar to a technique used
% by Viola~\cite{viola}.
% In particular, we expand $\rh$ to have $24\log n$ bits, but 
% rather than using all $24\log n$ bits of the $\rh$ function with each test in our doubling-binary search, we instead
% divide these bits into groups of $4$ bits each and just use the corresponding sets
% of bits in each comparison of a doubling-binary search. That is, for the first comparison in such a doubling-binary search, we use
% the first $4$ bits, in the second we use the second $4$ bits, and so on,
% so each comparison succeeds with probability at least $15/16$ and the result of each comparison is independent of any other. 
% The resulting version may require some backtracking, but uses only $6\log n$ tests to obtain a success probability of $1-n^{-3}$, see, e.g.,~\cite{dereniowski2025noisy,feige1994computing,eppstein2025computational}.

Thus, we have the following.
\begin{thm}
    Let $X$ and $Y$ be strings of length $l$ and $n$, respectively, with $l\le n$, from a finite alphabet,
    stored with separate parties.
    %, and let $k$ be a given upper bound on edit distance between $X$ and $Y$. 
    Then we 
    can compute the edit distance {$k$} between $X$ and $Y$ with an overhead of $O(n)$ computational operations for the two parties
    and a communication complexity of 
    $O(k^2\log n)$ bits whp.
\end{thm}
\begin{proof}
    The fast-sliding algorithm performs $O(k^2)$ doubling-binary searches, and by Lemma~\ref{lem:4bits}, each search requires $O(\log n)$ bits of communication and succeeds whp. 
    The total local computation consists of constructing the rolling hash in $O(n)$ and performing $O(k^2\log n)= O(n)$ hashes, since $k < \sqrt{n/\log_{|\Sigma |} n}$.
    Thus, the total communication is $O(k^2\log n)$ bits with $O(n)$ overhead.
\end{proof}

{Note that Alice and Bob not only discover the edit distance; they can reconstruct the path in their matrix $L$ that leads to such distance. This allows each party reconstruct the other string, with $O(n)$ additional time and space~\cite{Mye86}.}

\begin{comment}
\section{Communication-Efficient Hamming-Distance}
% \section{Set-based String Reconciliation}
% \subsection{Character String}
Consider two strings of equal length $X = x_1, x_2, \dots, x_n$ and $Y = y_1, y_2, \dots, y_n$, such that we want to compute the Hamming
distance between $X$ and $Y$.
In particular, let us assume that Alice is holding $X$
and Bob is holding $Y$, and we wish to compute
the Hamming distance 
between $X$ and $Y$ using the minimum amount of communication.
Our approach to solving this problem is based on the IBLT data
structure of Goodrich and Mitzenmacher~\cite{goodrich2015invertiblebloomlookuptables}, 
which we review next.
\end{comment}

\section{A Specialized Algorithm based on IBLT Chunking}

In many applications, the strings stored by Alice and Bob come from input distributions that {follow widely accepted statistical models \cite{Sha48,CT06}, as is the case of DNA and English text. Those models imply that all strings over some short length are unique whp.\footnote{{Note that this is not the case in some applications, for example those where repetitive sequences are handled \cite{Nav20}.}} 
%have some degree of entropy in practice, such as English text. 
In such cases, and if we know a priori an upper bound $k$ on the edit distance, we are able to possibly improve the communication complexity of the fast-sliding 
algorithm of the previous section, to $O(k\log^5 n)$ bits}. 
Concretely, we devise an efficient string reconciliation algorithm
that is based on \bemph{content-defined chunking}~\cite{rolling,fastcdc}.
We refer to this as our \bemph{IBLT-chunking} algorithm, since it also uses an invertible Bloom lookup table (IBLT)~\cite{goodrich2015invertiblebloomlookuptables}. 
This data structure is also used by Eppstein {\it et al.}~\cite{eppstein2011whatsthedifference} to efficiently
compute Hamming distance, but our usage is different.
We first give some background on IBLTs and then describe our algorithm.

% \subsection{IBLT Set Reconciliation}~\label{sec:iblt}
% We use an invertible Bloom lookup table (IBLT)  \cite{goodrich2015invertiblebloomlookuptables,eppstein2010straggler} as a black-box primitive to solve the \textit{set reconciliation} problem~\cite{eppstein2011whatsthedifference}, which is the task of synchronizing two sets stored at two nodes across a communication link. An IBLT stores elements by fixed-length binary keys and supports insertion,
% deletion, and recovery when the symmetric difference is small using XOR
% operations.

% Given two sets $S_A,S_B$ whose symmetric difference has size at
% most $d$, Alice and Bob construct IBLTs of appropriate size using the same parameters and hash functions, exchange their tables, and XOR each cell between the two tables. The resulting table holds the symmetric difference and
% can be decoded by the standard peeling procedure whp, provided that the number of elements remaining in the table is sufficiently small. With $O(d \log n)$ cells and appropriate parameters, Alice and Bob can reconcile their sets whp, in $O(n)$ time and $O(d \beta \log n)$ bits of communication, where $\beta$ is the size of the keys. The construction and decoding take linear time in the table size plus the number of inserted elements. The IBLT construction and decoding guarantees used here are in Appendix~\ref{app:iblt}.

\subsection{Invertible Bloom Lookup Table (IBLT)} \label{app:iblt}

IBLT  \cite{goodrich2015invertiblebloomlookuptables,eppstein2010straggler} is a space-efficient probabilistic data structure, based on Bloom filters, that supports various operations implemented by the XOR primitive. 
% Such operations include insertion, deletion, and lookup in $O(\lambda)$ time, where $\lambda$ is the number of distinct hash functions used, and listing in $O(t)$ time, where $t$ is the threshold representing the maximum number of elements in an IBLT that allows for successful decoding. IBLTs improve Bloom filters \cite{bloom1970} by allowing the number of key-value pairs to greatly exceed the threshold. Insertions and deletions into the IBLT are always guaranteed to succeed; however, while the current number of key-value pairs exceeds the threshold, other operations, such as lookup or listing elements, cannot be performed successfully. Once the number of elements in the IBLT falls below the threshold, the operations will succeed whp.
It can be used to solve the \textit{set reconciliation} problem~\cite{eppstein2011whatsthedifference}, which is the task of synchronizing two sets stored at two nodes across a communication link. IBLTs allow the parties to efficiently identify the items that are unique to each set so that both parties can update their local copies and obtain the other set as well. The main step in this process is to compute a \textit{set difference}, finding the set of elements that one party possesses but the other does not.
%Traditional methods often rely on transmitting entire datasets, which quickly becomes inefficient when the datasets are large but differ only slightly and updates are frequent.

Suppose we are given a set $S=\{s_1, s_2, \ldots, s_n\}$, which we want to store probabilistically in an IBLT table, $T$, with $m$ cells.
We assume that we have an \bemph{encoding function}, $\mathsf{encode}$, which encodes each element, $s_i$, into an associated
unique, fixed-length binary string key, $e_i$.
We also have $\lambda$ hash functions that map any key to $\lambda$ distinct locations, and we let $\mathrm{HashToIndices}(e_i,\lambda,m)$ denote
the set of $\lambda$ locations determined by these hash functions for the key $e_i$.
Each IBLT cell, $T[j]$, stores a \texttt{(keysum, hashsum)} pair such that \texttt{keysum} contains the bitwise-XOR 
(denoted $\oplus$) of all encoded keys $e_i$ that are mapped to $T[j]$ by one of $T$'s hash functions; 
\texttt{hashsum} contains the  XOR of all fingerprints or hashes $H(e_i)$, where $H(e_i)$ is an $\ell$-bit secondary hash of $e_i$, such as a checksum, or cryptographic hash, used to verify or detect errors in the decoding process.

%\subsubsection{Insertion, Deletion, and Set Encoding}
% For each element $x$ to be inserted or deleted, we proceed as in Algorithm~\ref{alg:insdel}. 
The IBLT supports both insertion and deletion of elements using XOR operations. Inserting an element corresponds to XOR-ing its encoded value into the table (at the $\lambda$ locations), while deleting it corresponds to {\em the same} operations.\footnote{This
method assumes each item is inserted or deleted at most once, which is true for sets.}
%This method is shown in Algorithms~\ref{alg:insert}.
%
%\begin{algorithm}[H]
%\caption{IBLT Insert/Delete $(T,s)$}
%\begin{algorithmic}[1]
%\STATE $e \leftarrow \mathsf{encode}(s)$
%\FOR{$j$ in HashToIndices($e,\lambda,m$)}
%    \STATE $T[j].\texttt{keysum} \gets T[j].\texttt{keysum} \oplus e$
%    \STATE $T[j].\texttt{hashsum} \gets T[j].\texttt{hashsum} \oplus H(e)$
%\ENDFOR
%\end{algorithmic}
%\label{alg:insert}
%\end{algorithm}
%
% The original method~\cite{eppstein2011whatsthedifference} performed deletions using the subtraction operator, but the XOR operator allows us to perform insertions and deletions identically while still enabling us to list entries whp, as shown in Corollary~\ref{cor:purity2}. These operations are used in Algorithms~\ref{alg:encode} and~\ref{alg:decode}.
%
To encode an entire set, $S$, we simply insert each of its elements.
%, as shown in Algorithm~\ref{alg:encode}.
%
%\begin{algorithm}[H]
%\caption{IBLT Encode Set $S$}
%\begin{algorithmic}[1]
%\FOR{$s_i \in S$}
%\STATE insert$(T,s_i)$
%\ENDFOR
%\end{algorithmic}
%\label{alg:encode}
%\end{algorithm}

For a set reconciliation protocol, Alice and Bob build respective tables, $T_A$ and $T_B$, by inserting all the elements of 
their sets, {$X$ and $Y$, respectively, into initially zeroed IBLTs of the same size using the same hash functions, and exchange the tables.
From these two IBLTs, $T_A$ and $T_B$, both Alice and Bob can compute an IBLT representing 
the symmetric difference between $X$ and $Y$.} This simple method
is shown in Algorithm~\ref{alg:subtract}.
The fact that this algorithm computes an IBLT representation of the symmetric difference of $X$ and $Y$ follows
immediately from the fact that $x\oplus x=0$ for any value $x$.

\begin{algorithm}[t]
\caption{IBLT Subtract ($T = T_A \oplus T_B$)}
\begin{algorithmic}[1]
\FOR{$i=0$ to $m-1$}
\STATE $T[i].\texttt{keysum} \gets T_A[i].\texttt{keysum} \oplus T_B[i].\texttt{keysum}$
\STATE $T[i].\texttt{hashsum} \gets T_A[i].\texttt{hashsum} \oplus T_B[i].\texttt{hashsum}$
\ENDFOR
\end{algorithmic}
\label{alg:subtract}
\end{algorithm}

\subsubsection{Listing Set Entries}
If an IBLT is not ``too'' full, we can list out its entries.
The listing of entries or decoding is a destructive $O(m)$-time procedure that recovers all keys in the IBLT or reports
that the IBLT is too full to achieve this successfully.
Say that a cell $T[i]$ is \bemph{pure} if
\[
H(T[i].\texttt{keysum}) = T[i].\texttt{hashsum},
\]
which clearly holds for a cell holding exactly one key but could admittedly also hold if there is a checksum
collision. 
Note that
if a cell is pure, we can recover its key by extracting the value from $T[i].\texttt{keysum}$ and then deleting the element
from the table, which is also referred to as \bemph{peeling} they key from the corresponding $\lambda$ hashed cells. 
This may make other cells pure, which can eventually allow us to recover all the values stored in the IBLT.
Algorithm~\ref{alg:decode} shows this process.

\begin{algorithm}[t]
\caption{IBLT Decode $T$}
\begin{algorithmic}[1]
\WHILE{$\exists~i$, s.t. $H(T[i].\texttt{keysum}) = T[i].\texttt{hashsum}$}
    \STATE $s\gets \mathsf{decode}(T[i].\texttt{keysum})$
    \STATE Output $s$
    \STATE delete$(T,s)$
\ENDWHILE

\FOR{$i=0$ to $m-1$}
    \IF{$T[j].\texttt{keysum} \neq 0$ OR $T[j].\texttt{hashsum} \neq 0$}
    \STATE \textbf{return} FAIL
    \ENDIF
\ENDFOR
\STATE \textbf{return} SUCCESS
\end{algorithmic}
\label{alg:decode}
\end{algorithm}

\begin{lemma} \label{lem:purity1}
    The probability that a cell containing $t \ge 2$ keys produces a false positive on the purity test is $2^{-\ell}$.
\end{lemma}
\begin{proof}
    Let $H$ be a uniform random hash function that produces $\ell$-bit hash values.
    Suppose a cell contains $t \ge 2$ keys, with \texttt{keysum} $=x_1 \oplus x_2 \oplus \ldots \oplus x_t$ and \texttt{hashsum} $=H(x_1) \oplus H(x_2) \oplus \ldots \oplus H(x_t)$. Define a random variable $Z=H(\texttt{keysum})\oplus \texttt{hashsum}=H(x_1 \oplus \ldots \oplus x_t)\oplus H(x_1) \oplus \ldots \oplus H(x_t)$. Because $H$ is uniform and random, $H(x_1 \oplus \dots \oplus x_t)$ is independent of $H(x_1), \dots, H(x_t)$ and each value is uniformly distributed over $\{0,1\}^\ell$. Thus, the probability of failure is $2^{-\ell}$.
\end{proof}

\begin{corollary} \label{cor:purity2}
    In an IBLT of $m$ cells, the probability that the purity test yields a false positive for any cell is at most $m 2^{-\ell}$.
\end{corollary}

\begin{proof} %\bemph{(of Corollary~\ref{cor:purity2})}
    Apply the union bound to the failure probability from Lemma~\ref{lem:purity1}.
\end{proof}

In our setting, we want to recover a symmetric difference of size at most $d$ with high probability, even if $d$ is a constant.
We use an IBLT of size $m = c \cdot d\log n$ for a $c\ge 1$, and $\lambda$
being $\Theta(\log n)$, so that by
an analysis similar to that of Goodrich, Kitagawa, and Mitzenmacher~\cite{goodrich2025parallel}, when the current number of elements in the IBLT is at most $d$, we can successfully list all entries of $T$ whp:

\begin{thm}
\label{thm:whp}
    Let $T$ be an IBLT with  $m=cd\log n$ cells and 
    $\lambda=(c/2)\log n$ hash functions, for $c \ge 4$. 
    If $T$ stores at most $d$ elements, then it can be successfully decoded by the peeling algorithm with 
    probability at least $1-1/n^{c/2-1}$.
\end{thm}

\begin{proof} %\bemph{(of Theorem~\ref{thm:whp})}
    % $P[X_{i,j}=0]=(1-\frac{1}{m})^{(d-1)\lambda}=(1-\frac{1}{cd\log n})^{(d-1)c\log n} \geq (1-\frac{1}{cd\log n})^{cd\log n} \approx e^{-1}$
    Let $X_{i}$ be a random variable that is $1$ if the $i$-th element has no pure cell and is $0$ otherwise.
    Since there are at most $d$ elements in $T$ and it has $m=cd\log n$ cells and 
    $\lambda=(c/2)\log n$ hash functions, at least half the cells in the IBLT are empty. Thus, any one of the $\lambda$ hash functions will map an element
    to what would otherwise be an empty cell with probability at least $1/2$.
    Since the $\lambda$ hash functions are independent, this implies that
    \[
    \Pr(X_i=1) \le \frac{1}{2^\lambda} = \frac{1}{n^{c/2}}.
    \]
    Thus, by a union bound, the probability that every element in $T$ has a pure cell is at least
    $1-1/n^{c/2-1}$.
    %the probability that none of the $\lambda$ hashes for an are pure is 
    %$X_{i,j}=1$ iff the $j$-th hash of the $i$-th item goes to
    %a cell that is not \textit{pure} and $0$ otherwise. 
    %Since the IBLT has $m = cd\log n$ cells and $\lambda = c\log n$ hash functions, we have $P[X_{i,j}=0]=(1-\frac{1}{cd\log n})^{(d-1)c\log n} \geq e^{-1}$. Thus, each hash of $\lambda$ has constant probability. Since the hash functions are independent, the probability that none of the $\lambda$ hashes are pure is 
    %$(1-e^{-1})^{c\log n}=n^{c\log (1-e^{-1})}=n^{-c'}$ for $c'=-c\log (1-e^{-1})$. By a union bound, over all $d$ items, the probability that any item fails all $\lambda$ hashes is no more than 
    %$d\cdot n^{-c'}\leq n \cdot n^{-c'} = n^{-(c'-1)}$, which remains inverse-polynomial in $n$.
\end{proof}

To sum up, suppose two sets, $X$ and $Y$, of size $O(n)$, are held by Alice and Bob respectively, which are assumed to differ by at most $d$ elements. Each party independently constructs an IBLT, $T_A$ and $T_B$ of $m = O(d\log n)$ cells, and inserts all elements into the table. Even if either $T_A$ and $T_B$ is too full to decode individually, there will be pure cells in $T_A \oplus T_B$, i.e., containing only one key. That is, with properly chosen parameters, Alice and Bob can reconcile their sets whp, in $O(n)$ time and $O(d\beta\log n)$ bits of communication, where $\beta>0$ is the size of the keys, by exchanging their tables $T_A$ and $T_B$ and performing the decoding algorithm on $T_A \oplus T_B$. %1.5d
%Let $D_{A-B}$ denote the set of elements that appear in $A$ but not in $B$, and similarly for $D_{B-A}$. The goal is to allow Alice and Bob to efficiently reconcile their differences while minimizing communication costs.

\subsection{Our Algorithm}

\bemph{Content-defined chunking} (CDC)~\cite{rolling,fastcdc} divides a string, $X=x_1x_2\ldots x_n$,
into a set, $\mathcal{C}$, of contiguous substrings, called \bemph{chunks},
$X_1$, $X_2$, $\ldots$, $X_r$, so that $X=X_1 || X_2 || \cdots || X_{r}$, where $||$ denotes concatenation.
This is achieved by defining a minimum-size parameter, $s$, and a windowed rolling hash function, $h$, with a window
size, $w$, and subdividing the string, $X$, by a simple algorithm that computes $h$ for each substring of $X$ of length $w$,
starting from the beginning of $X$. At each position where $h(x_i\ldots x_{i+w-1})=0$, if the length of the current substring being
processed is at least $s$, then we cut $X$ at the index $i+w-1$ to form the next chunk.

In the edit distance algorithm of this section, we desire that the chunks in $\mathcal{C}$ are unique and not too long, which we 
can characterize in terms of a parameter $d_X$: the length of the longest repeated string in $X$. Note that $d_X$ is computed in linear time using suffix trees. The following observation should be obvious.
%with respect to the suffix tree $T_X$  for the string $X$:
% $X_i$, is a prefix of any other chunk, $X_j$, $i\not=j$, and that the chunks cover $X$.
%let $d(T_X)$ denote the maximum string depth of an internal node in $T_X$.

 \begin{lemma}
     Let $\mathcal{C}=\{X_1,X_2,\ldots,X_r\}$ 
     be a set of chunks for the string $X=X_1 || X_2 || \cdots || X_{r}$. If
     $|X_i| > d_X$ for $i=1,2,\ldots,r$, then each chunk in $\mathcal{C}$ is unique.
 \end{lemma}
%\begin{proof}
%For each chunk, $X_i\in \mathcal{C}$, let $v_i$ be the locus of $X_i$ in $T_X$. Note that each such $v_i$ must exist, since the chunks in $\mathcal{C}$ occur in $X$.
%Since each $X_i$ has length at least $d(T_X)+1$, $v_i$ must be a leaf, therfore $X_i\in \mathcal{C}$ is unique.
%\end{proof}

Szpankowski~\cite{szpankowski1993generalized} shows that the value of $d_X$ for strings $X$ that can be modeled
as pseudo-random sequences, such as those determined by Bernoulli or Markov processes, is $O(\log |X|)$ whp.
Note that $d_X$ is always at least $\log_{|\Sigma|} |X|$. 
\ifFull
Appendix~\ref{app:practice} discusses practical values.
\else
Appendix A  of the extended version \cite{GNT26} discusses practical values.
\fi

Before running the string reconciliation protocol, Alice and Bob first compute and exchange $d_X$ and $d_Y$, so that both use the same minimum chunk size, $s=\max\{d_X,d_Y\}+1$, and the same CDC window size, $w=s$, for their respective chunking processes. 
Thus, no chunk boundary is allowed before it has accumulated $s$ characters.
For the IBLT construction, we choose a confidence bound, $n<\eta\le n^c$, for some
constant $c\ge 1$, where we desire
our protocol to succeed with probability at least $1-1/\eta$ (i.e., whp).
%For example, we could choose $\eta=100$ for our protocol to succeed with 99\% confidence or we could choose $\eta=n^c$, for some constant $c\ge 1$, for our protocol to succeed with high probability.
%Accordingly, in the edit distance algorithm we describe in this section, we define a window size, $w$, and minimum-size,
%$s$, for the rolling hash function, $h$, for $s=w\ge d(T_X)+1$. 
% Moreover, we recommend only applying our algorithm in the common case where $d(T_X)$ is $O(\log |X|)$. 

\begin{lemma}
With high probability,
the maximum chunk size in $\mathcal{C}(X)$ (resp., $\mathcal{C}(Y)$) is 
$O(s+\log |X|)$ (resp., $O(s+\log |Y|)$).
\label{lem:log}
\end{lemma}
\begin{proof}
   W.l.o.g., let us focus on Alice's set of chunks, $\mathcal{C}(X)$.
    By definition, every substring of $X$ with length at least $s$ is guaranteed to be unique.
Suppose that the uniformly random rolling hash function $h$ outputs $b$ bits. 
The probability of a cut at any window $w$ is $p= Pr[h(x_i\ldots x_{i+w-1})=0]=1/{2^b}$. 
     Suppose a chunk of length $l>w+t$ starts at position $i$. Then the chunk contains at least $t$ distinct windows of size $w$, with starting positions in the range $[i,i+t-1]$. For the chunk to have length greater than $w+t$, none of these $t$ windows can produce a cut.
    Therefore, $\Pr[l>w+t] \leq (1-p)^t \leq e^{-pt}$.
    Let {$t=(c/p) \log |X|$} %$(1/\varepsilon)$
    for a constant $c \geq 2$, 
then {$Pr[l>w+t] \leq e^{-pt} =e^{-c \log |X|} = %(1/\varepsilon)} = 
1/|X|^c$}. 
    By a union bound over all $|X|$ possible starting positions, the probability any chunk exceeds $w+t$ is at most $1/|X|^{c-1}$. Since {$w=s$} %$w=d(T_X)+1=O(d(T_X))$ 
    and $t=O(\log |X|)$ for constant hash output size $b$ (which makes $p$ constant), with high success probability, at least $1-1/|X|^{c-1}$ where $c \geq 2$, the maximum chunk size is {$O(s+\log|X|)$}. %$O(d(T_X)+\log|X|)$.
\end{proof}

Let us further assume that we have an upper-bound estimate, $k\ge1$, 
for the edit distance between the strings, $X$ and $Y$, held respectively
by Alice and Bob (we later lift this assumption). 
%For example, $k$ could be found by
%using a doubling estimation strategy, as in the
%protocol of Eppstein {\it et al.}~\cite{eppstein2011whatsthedifference}.
Then we create an IBLT of appropriate size to allow listing its elements 
with probability at least $1-1/\eta$, for
our desired confidence parameter of successful decoding, $n<\eta\le n^c$, 
if there are most $\kappa$ elements, i.e.,
an IBLT of
$O(\kappa\log \eta)$ cells, with each cell storing a \texttt{keysum} and \texttt{hashsum} of $\Theta(\log n)$ bits (here, $\kappa$ is a function of $k$ to be defined later).
Alice then stores a suitably defined set, $S(X)$, in her IBLT of this chosen size such that,
given a 
string, $X$, and a set of chunks,
$\mathcal{C}(X)=\{X_1, X_2, \ldots, X_r\}$, for $X=X_1 || X_2 || \cdots || X_{r}$, her IBLT is defined with respect to the set $S(X)$ defined for her
chunks.
In particular,
she encodes $X=x_1x_2\ldots x_n$ as a set of pairs of checksum hashes of consecutive chunks, 
\[
S(X)=\{ H(X_i)||H(X_{i+1}) \ |\ X_i\in \mathcal{C}(X), i=1,\ldots,r-1 \},
\]
where $H(X')$ is a checksum hash with $\Theta(\log n)$ bits for the chunk $X'$, so that, whp, all the checksum hashes are unique, given that all
the chunks are unique.
Thus, each key associated with consecutive chunks is the concatenation of their fixed-length hash values and is itself a fixed-length bit string.
Likewise, Bob encodes a similar set of pairs of chunk checksums, $S(Y)$, and stores these in a similarly-defined IBLT of the same size.
Alice and Bob
then exchange their IBLTs and perform the peeling algorithm to find the set difference.
% See Figure~\ref{fig:chunking}.

% \begin{figure}[hbt]
%     \centering
%     \includegraphics[width=.85\columnwidth]{figs/chunking_fig.pdf}
%     \caption{Our IBLT-chunking protocol for string reconciliation.}
%     \label{fig:chunking}
% \end{figure}

For each differing chunk pair that either Alice or Bob decodes, they can determine if it is one of their chunk pairs or a chunk pair from the
other party.
Further, after Alice and Bob have respectively decoded the set-difference IBLT and determined which chunk pairs
belong to $X$ and which belong to $Y$,
Alice and Bob then each send the other their
corresponding chunk pairs. 
%chunk-checksum pair in $S(X)$ or $S(Y)$, respectively. That is,
%Alice sends Bob $(X_1,X_2)$ and Bob sends Alice $(Y_1,Y_2)$.
%Then Alice and Bob send the other the corresponding chunk pairs for each of their chunk-checksum pairs that were 
%in the symmetric difference set decoded from the IBLT.
That is, for each chunk-checksum pair, $H(X_i)||H(X_{i+1})$, of hers that Alice decodes from the IBLT as being a part
of the symmetric difference, Alice sends Bob
the pair $(X_i,X_{i+1})$. Likewise, Bob sends a similar set of chunk pairs to Alice.

\begin{lemma}
    Let $k$ be an upper bound on
    the edit distance between two strings, $X$ and $Y$, each of length $\Theta(n)$, defined over an alphabet $\Sigma$,
    let $d_X$ and $d_Y$ be the length of their longest repeated substring,
    respectively, and $s=\max\{d_X,d_Y\}+1$. 
    Then our IBLT-chunking algorithm finds the set of $O(k)$ differing chunk pairs between $X$ and $Y$ with probability at
    least $1-1/\eta$, for a desired confidence bound, $n<\eta\le n^c$,
    for a constant $c\ge 1$,
    using 
    $
    O(k s^2 \log n((s+\log n)\log |\Sigma| + \log \eta \log n))
    $
    bits of communication.
\end{lemma}
\begin{proof}
     By the above discussion and Lemma~\ref{lem:log}, every chunk has length $O(s+\log n)$ characters whp. Encoding each character requires $\lceil\log |\Sigma|\rceil$ bits, so each chunk will have size $O((s+\log n)\log |\Sigma|)$ bits whp. We show that the number of differing chunks to exchange is $\kappa = O(k s^2 \log n)$ whp. Plus, the IBLT needs  $O(\kappa \log \eta)$ cells of $O(\log n)$ bits each to decode with probability at least $1-1/\eta$. 

% redundant: if $k$ is an upper bound on the edit distance between $X$ and $Y$, there will be $O(k)$ differing chunk pairs between $S(X)$ and $S(Y)$, 
Since any edit between $X$ and $Y$ will change at most $w$ consecutive rolling hashes, and since $s=w$, those positions span at most two consecutive chunks. The problem is that, because of the edit, the rolling hashes on $X$ and $Y$ may yield zero at different positions, and since we do not consider other zeros before scanning the next $s=w$ symbols, the next zeros considered on $X$ and $Y$ may differ again, even if the preceding windows are identical by then.
This propagation, however, ends soon whp: Assume both processes are synchronized up to now. Once an edit occurs, and after passing the minimum $s$ symbols if required, we can take the distance to the next zero of the rolling hash in $X$ as a random variable $\mathcal{X}_1$, and in $Y$ as a random variable $\mathcal{Y}_1$, both geometrically distributed with mean $2^b$. If $\mathcal{X}_1=\mathcal{Y}_1$, then they have resynchronized. Otherwise, assume w.l.o.g. $\mathcal{X}_1 < \mathcal{Y}_1$, and call $\mathcal{X}_i$ and $\mathcal{Y}_i$ the number of rolling hash probes in the $i$th chunk. The process can resynchronize in two ways (see Figure~\ref{fig:resync}):
\begin{enumerate}
\item For some $\tau\ge 1$,
$\mathcal{X}_1+s+\cdots+\mathcal{X}_\tau+s+\mathcal{X}_{\tau+1} \ge \mathcal{Y}_1+s+\cdots+\mathcal{Y}_\tau+s = j$.
\item For some $\tau\ge 1$,
$j = \mathcal{X}_1+s+\cdots+\mathcal{X}_\tau+s \le \mathcal{Y}_1+s+\cdots+\mathcal{Y}_\tau$.
\end{enumerate}

\begin{figure}[t]
\centering
\includegraphics[width=0.85\textwidth]{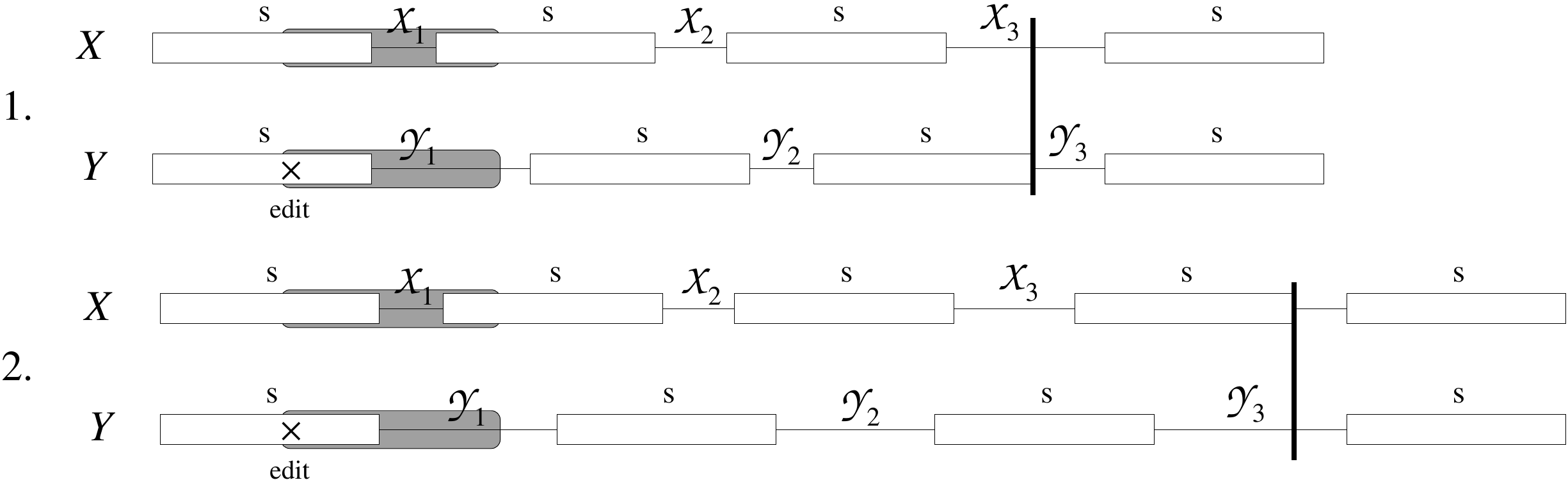}
\caption{Cases of resynchronization. Chunks are a rectangle (where they cannot end) followed by a line (of lengths $\mathcal{X}_i$ or $\mathcal{Y}_i$) where they end once the rolling hash is $0$. Hashes can differ in the grayed areas only. Thick vertical lines show the synchronization point.}
\label{fig:resync}
\end{figure}

In both cases, the two windows ending at position $j$ are identical (far away
from the edit position) and of length at least $s$, so both chunks are ended at $j$.
Let $\mathcal{Z}_k = \sum_{i=1}^k \mathcal{X}_i-\mathcal{Y}_i$: we resynchronize as soon as (1) $\mathcal{Z}_\tau+\mathcal{X}_{\tau+1} \ge 0$ or (2) $\mathcal{Z}_\tau \le -s$ for some $\tau$. 
%Once both processes skip other $s$ characters after $\mathcal{X}_1$ and $\mathcal{Y}_1$, they process $\mathcal{X}_2$ and $\mathcal{Y}_2$ further symbols after reaching zeros again. If $\mathcal{X}_1 + s + \mathcal{X}_2 \ge \mathcal{Y}_1 +s ~(\pm 1)$, they resynchronize at $\mathcal{Y}_1 +s ~(\pm 1)$. If $\mathcal{X}_1+s \le \mathcal{Y}_1~(\pm 1)$, they also resynchronize, at $\mathcal{X}_1+s~(\pm 1)$. In general, we have geometric variables $\mathcal{X}_i$ and $\mathcal{Y}_i$, define $\mathcal{Z}_k = \sum_{i=1}^k \mathcal{X}_i-\mathcal{Y}_i$ and resynchronize as soon as $\mathcal{Z}_\tau+\mathcal{X}_{\tau+1} \ge 0$ of $\mathcal{Z}_\tau \le -s$ for some $\tau$. 
\ifFull
In Appendix~\ref{app:resync} 
\else
In Appendix B of the full version~\cite{GNT26}
\fi
we %give more details and 
show that this occurs whp after $\Theta((s^2/4^b) \log n)$ steps, and on average after $\Theta(s^2/4^b)$ steps. For insertions or deletions, the positions are shifted by $\pm 1$ but the result is identical.

%, at most two consecutive chunks can differ for each edit. Therefore, each edit of a string can change at most $O(1)$ chunks 
%when doing CDC chunking (i.e., chunking quickly
%    resynchronizes after an edit).

Thus, if $k$ is an upper bound on the edit distance between $X$ and $Y$, then at most $\kappa = O(k s^2 \log n)$ chunks
 differ, whp, between $\mathcal{C}(X)$ and $\mathcal{C}(Y)$; hence, $O(k s^2 \log n)$ chunk-checksum pairs 
     are different between $S(X)$ and $S(Y)$ whp.
    Thus, after the decoding of the IBLT occurs, Alice and Bob each send the other lists of $O(k s^2 \log n)$ chunk pairs,
    each of at most $O((s+\log n)\log |\Sigma|)$ bits.
\end{proof}

% Thus, the symmetric difference IBLT will be of small size if $d(T_X)$ and $d(T_Y)$ are small. 

In addition, 
the information exchanged is 
sufficient for each party to reconstruct the other party's string.
W.l.o.g., let us focus on Bob, who wishes to reconstruct $X$ by ``linking'' all of Alice's chunk pairs. Also, 
with probability at least $1-1/\eta$, we can assume all the elements in the symmetric difference were decoded correctly
and that all the chunk checksums are unique.
For instance, the chunks in $\mathcal{C}(X)$ (resp., $\mathcal{C}(Y)$) are unique by construction;     
    hence, the chunk-checksum pairs in $S(X)$ and $S(Y)$ are also respectively unique whp.
Further, {assume for simplicity that Alice always sends Bob} her first chunk pair $(X_1,X_2)$, so Bob can determine if their first chunk pairs match or not. 
In either case,
he can determine the first two chunks in $\mathcal{C}(X)$.
Also, note that Bob has received from Alice each chunk pair, $(X_j,X_{j+1})$, such that 
$H(X_j)||H(X_{j+1})$ is a chunk-checksum pair in the symmetric difference of $S(X)$ and $S(Y)$ decoded from the IBLT.
Suppose $(X_{i-1},X_{i})$ is the last pair of chunks in $\mathcal{C}(X)$ that Bob has discovered; hence, he is interested
in discovering the next pair, $(X_i,X_{i+1})$.
\begin{itemize}
\item {Case 1.} There is a chunk pair, $H(X_i)||H(X')$, in the IBLT symmetric difference. Since chunks are unique, this implies
       $(X_i,X_{i+1})=(X_i,X')$. Moreover, in this case, Alice has sent Bob the pair, $(X_i,X')$; hence, Bob can determine
       $(X_i,X_{i+1})$.
\item {Case 2.} There is no chunk pair, $H(X_i)||H(X')$, in the IBLT symmetric difference. Then there is a chunk-checksum pair, $H(X_i)||H(X')$, that is also in $S(Y)$; i.e.,
       $(X_i,X')=(Y_j,Y_{j+1})$, for some $j>1$. Since chunks are unique, this implies $(X_i,X_{i+1})=(Y_j,Y_{j+1})$; hence, Bob can determine $(X_i,X_{i+1})$ by looking up his chunk, $Y_{j+1}$, given  that he has already determined that
       $X_i$ matches $Y_j$.
\end{itemize}

%Thus, we have the following.
 
\begin{corollary}
    The exchanged information is sufficient for Bob (resp., Alice) to reconstruct $X$ (resp., $Y$) with probability at least $1-1/\eta$. 
\end{corollary}

Therefore, we have the following result, which yields $O(k\log^5 n)$-bit 
communication complexity whp in the discussed models where $d_X$ and $d_Y$ are $O(\log n)$.

\begin{thm}
\label{thm:prob}
    Let $k$ be an {a-priori-known} upper bound on
    the edit distance between two strings, $X$ and $Y$, of length $\Theta(n)$, each defined over {an alphabet $\Sigma$},
held by two separate parties, Alice and Bob,
    let $d_X$ and $d_Y$ be the lengths of their longest repeated substrings, and $s = \max \{d_X,d_Y\}+1$. 
    Then our IBLT-chunking algorithm allows Alice and Bob to determine the 
other party's string in $O(n)$ time
    using 
    \[
    O(k s^2 \log n(s\log |\Sigma| +\log \eta\log n))
    \]
bits of communication with probability at least $1-1/\eta$.
\end{thm}

 {Once both Alice and Bob have the chunks of both strings, they can compute their edit distance in $O(n+k^2)$ additional time \cite{landau1998incremental}.} 
 %Not better in the worst case: For this, they just need to sum up the edit distances between the $O(k)$ differing chunks. Using an algorithm that computes a distance $d$ in time $O(d^2)$ \cite{}, this takes extra time $O(k^2)$.}
 %

Thus, when $k = O(\sqrt{n})$ and $d_X$ and $d_Y$ are $O(\log n)$, we can also compute the precise edit distance within this cost, which may outperform our general fast-sliding algorithm.
%in the likely case were $k$ is relatively small compared to $n$ and $X$ and $Y$ are real-world strings
%where $d(T_X)$ and $d(T_Y)$ are also small, say, $O(\log n)$ or $O(\log^2 n)$, we can significantly improve on the communication efficiency of our
%fast-sliding algorithm, which works for more general scenarios.
%{For example, we have the following simplified results.}
%
%\begin{corollary}
%{Let $k$ be an a-priori-known upper bound on
%    the edit distance between two strings, $X$ and $Y$, of 
%length $\Theta(n)$, defined over a finite alphabet, held by two separate parties, Alice and Bob,
%    and let $T_X$ and $T_Y$ be
%    the suffix trees for $X$ and $Y$, respectively,
%such that $d(T_X)$ and $d(T_Y)$ are $O((\log \eta)\log n)$. 
%Then our IBLT-chunking algorithm allows Alice and Bob to determine the other party's string in $O(n)$ time with probability at
%    least $1-1/\eta$, for a desired confidence bound, $0<\eta\le n^c$,
%    for a constant $c\ge 1$,
%    using $O(k(\log \eta)\log n)$ bits of communication.}
%\end{corollary}

%\begin{corollary}
%{Let $k$ be an a-priori-known upper bound on
%    the edit distance between two strings, $X$ and $Y$, of 
%length $\Theta(n)$, defined over an alphabet 
%    of size $O(n^c)$ for any constant $c$, and let $T_X$ and $T_Y$ be
%    the suffix trees for $X$ and $Y$, respectively,
%such that $d(T_X)$ and $d(T_Y)$ are $O(\log n)$. 
%Then our IBLT-chunking algorithm allows Alice and Bob to determine each other party's string in $O(n)$ time with high probability,
%    using $O(k\log^2 n)$ bits of communication.}
%\end{corollary}
%
{To conclude, we note that we can obtain a similar result {\em without} knowing an \textit{a priori} bound $k$ on the edit distance, by slightly increasing the time overhead.}

\begin{corollary}
    {The same result of Theorem~\ref{thm:prob}  can be obtained if $k$ is the actual edit distance between $X$ and $Y$, a priori unknown, with $O(n\log k)$ time overhead.}
\end{corollary}
\begin{proof}
{Apply Theorem~\ref{thm:prob} with upper bound $k'=1,2,4,8,\ldots$ until it succeeds. With probability at least $1-1/\eta$, the try with $k \le k' < 2k$ will succeed, and each party will obtain the other's string. The total communication cost is as in Theorem~\ref{thm:prob}, replacing $k$ by $\sum_{j=0}^{\lceil \log_2 k\rceil} 2^j = O(k)$, and the time overhead is $O(n\log k)$ because we process the chunks $O(\log k)$ times, once per value of $k'$.}
\end{proof}

\ifFull
\section{Conclusion}

We have introduced simple low-overhead communication efficient algorithms for string reconciliation between two parties that separately store two strings of size $\Theta(n)$.
Our general method, the fast-sliding algorithm, runs with linear time complexity for the parties and exchange
$O(k^2\log n)$ bits whp, where $k$ is the edit distance between both strings.
We also gave an IBLT-based method based on content-defined chunking
that also has linear time complexity for the parties and, for typical real-world strings whose longest repeated substrings are of length $O(\log n)$, exchange $O(k\log^5 n)$ 
%$O(k(\log \eta)\log n)$ 
bits whp.
% %, with probability at least
% %$1-1/\eta$, for a confidence parameter, $0<\eta\le n^c$, for 
% %any fixed constant, $c\ge 1$.

Directions for future work include exploring efficient low-overhead ways of performing string reconciliation for compressed strings, particularly in the scenario of highly repetitive string collections, which do not follow typical statistical models but can be sharply compressed with dictionary methods. 
We will also consider other similarity models beyond edit distance, like allowing substring reversals and block moves.
\fi

\begin{comment}
\subsection*{Acknowledgements}
This research supported in part by NSF grant 2212129.
\end{comment}

\bibliographystyle{IEEEtranS}
\bibliography{ref}

\ifFull
\appendix

\section{Lengths of Unique $q$-grams in Practice} \label{app:practice}

In practice, many studies involving DNA strings select $q$ in the range of 20-40 base pairs to achieve $q$-grams (aka $q$-mers) that are mostly unique~\cite{ponsero2023comparison}. This range balances the high probability of representing unique genomic sequences while reducing potential sequencing errors~\cite{rahman2018association}. For example, $q=32$ captures $85.7\%$ of unique sequences in the human genome~\cite{shajii2016fast}, $q=25$ suffices for majority of $q$-mers to be unique in mammalian genomes, many of which are already unique for $19$-mers~\cite{Zentgraf2024Swiftly}, and $q=21$ distinguishes many eukaryote, bacteria, and archea species~\cite{bussi2021large}. 
For natural language, there is no formal source guaranteeing an exact number of words or characters needed for uniqueness. However, analyses of large corpora~\cite{brants2006web1t, 1084beece4b5432ea2dd848cdc1a765a} show that $5$-grams are sufficient for most such sequences of words to be unique, supporting the use of 4-6 word sequences in practical NLP applications.
For example, the Indiana University test for word-for-word plagiarism is a matching sequence of seven or more words between two documents~\cite{IUtrust}.
In both applications, the choice of $q$-gram length is related to the depth, $d(T_X)$, in practice. 

\section{Probability of Resynchronizing} \label{app:resync}

Consider two sequences, $\mathcal{X}_1,\mathcal{X}_2,\ldots$ and $\mathcal{Y}_1,\mathcal{Y}_2,\ldots$ of independent geometric random variables with success probability $p = 2^{-b}$. Let us define another random variable as their cumulative difference: $\mathcal{Z}_k = \sum_{i=1}^k \Delta \mathcal{Z}_i$, where $\Delta \mathcal{Z}_i = \mathcal{X}_i - \mathcal{Y}_i$. Note $\Delta \mathcal{Z}_i$ is symmetric around zero, with mean $\mathbb{E}[\Delta \mathcal{Z}_i] = 0$ and variance $\sigma^2 = \text{Var}(\mathcal{X}_i) + \text{Var}(\mathcal{Y}_i) = 2 (1-p)/p^2 = 2(2^{2b} - 2^b)$. It also holds $\mathbb{E}[\mathcal{Z}_k]=0$.

%Let us assume $\mathcal{X}_1 < \mathcal{Y}_1$ (i.e., $\mathcal{Z}_1 < 0$; the other case is symmetric). The process can resynchronize in two ways, see Figure~\ref{fig:resync}:
%\begin{enumerate}
%\item For some $\tau\ge 1$,
%$\mathcal{X}_1+s+\cdots+\mathcal{X}_\tau+s+\mathcal{X}_{\tau+1} \ge \mathcal{Y}_1+s+\cdots+\mathcal{Y}_\tau+s = j$.
%\item For some $\tau\ge 1$,
%$j = \mathcal{X}_1+s+\cdots+\mathcal{X}_\tau+s \le \mathcal{Y}_1+s+\cdots+\mathcal{Y}_\tau$.
%\end{enumerate}
%
%\begin{figure}[t]
%\centering
%\includegraphics[width=0.8\textwidth]{resync.pdf}
%\caption{Cases 1 and 2 for resynchronization. Chunks start with a rectangle (where they cannot end) followed by a line (of lengths $\mathcal{X}_i$ or $\mathcal{Y}_i$) where they end if the rolling hash reaches zero. Thick vertical line shows the synchronization point.}
%\label{fig:resync}
%\end{figure}
%
%In both cases, the two windows ending at position $j$ are identical (far away
%from the edit position) and of length at least $s$, so both chunks are ended at $j$.
%
%Therefore, the process terminates at the first index $\tau \ge 1$ where either (1) $\mathcal{Z}_\tau+\mathcal{X}_{\tau+1} \ge 0$, or (2) $\mathcal{Z}_\tau \le -s$. We aim to find a sufficiently large $m$ such that the process terminates before $m$ steps (i.e., $\tau \le m$) with probability $\ge 1-1/n^c$.

We aim to find a sufficiently large $m$ such that $\tau \le m$ with probability $\ge 1-1/n^c$, where $\tau$ is defined as
\[
\tau = \min \left\{ i \ge 1 : \left( \mathcal{Z}_i + \mathcal{X}_{i+1} \ge 0 \right) \lor \left( \mathcal{Z}_{i+1} \le -s \right) \right\}.
\]

Let us divide the range $[1..m]$ into consecutive, non-overlapping blocks of length $\ell = \lceil s^2 / \sigma^2 \rceil$. At any given step $i$, the walk remains between boundaries iff $-s < \mathcal{Z}_i < -\mathcal{X}_{i+1}$. Let us pessimistically expand the boundaries to $(-s,0)$; the final result is not affected. For a block starting at $i$, we define $p_0 > 0$ as the minimum probability of escaping the boundaries across all possible starting positions within boundaries. To make $p_0$ independent of $s$, we set
\[ p_0 = \inf_{s \ge 1} \min_{x \in (-s, 0)} \Pr( \tau \le i+\ell, \mathcal{Z}_i = x)\]
(note that $p_0$ is independent of $i$ because the increments $\Delta\mathcal{Z}_i$ are memoryless).
Choosing $\inf_{s \ge 1} \min_{x \in (-s,0)}$ is pessimistic, but good enough: Because a standard random walk has a guaranteed positive chance to diffuse across a distance $s$ in $s^2/\sigma^2$ steps \cite{lawler2010random}, $p_0$ is a fixed positive constant that never drops to zero.

Therefore, across $t = \lceil m/\ell\rceil$ blocks, the probability of remaining
within bounds is $\Pr(\tau > m) \le (1-p_0)^t$. Since $t \ge m\sigma^2/s^2$,
the failure probability is $\Pr(\tau > m) \le (1-p_0)^{m\sigma^2/s^2}$, and
this is at most $1/n^c$ whenever
\[ m ~\ge~ \frac{s^2 \cdot c \log n}{\sigma^2 \log\frac{1}{1-p_0}}
   ~=~ \Omega((s^2/4^b) \log n).\]

To bound the expectation of $\tau$, note that, because up to $\tau-1$ the process is constrained
within the band $(-s,0)$, it always holds
$\mathcal{Z}_k^2 \le s^2$ for $k<\tau$, and thus
$\mathbb{E}[\mathcal{Z}_{\tau-1}^2] \le s^2$. On the other hand,
since the $\Delta\mathcal{Z}_k$ are i.i.d., the variance grows linearly with time, $\text{Var}(\mathcal{Z}_{\tau-1}) = \mathbb{E}[\mathcal{Z}_{\tau-1}^2] = (\tau-1)\sigma^2$ \cite{durrett2019probability}.
By the linearity of expectation for stopped processes \cite{durrett2019probability}, we then have $\mathbb{E}[\mathcal{Z}_{\tau-1}^2] = \mathbb{E}[\tau-1] \cdot \sigma^2$. Thus, $\mathbb{E}[\tau] \le 1+s^2/\sigma^2 = O(s^2/4^b)$.

\end{document}